\documentclass[9.5pt,journal,letterpaper,twocolumn]{IEEEtran}

\usepackage{ifthen}
\usepackage{color}
\newboolean{showeditcomments}
\setboolean{showeditcomments}{true} 
\newcommand{\editcomment}[2]
          {\ifthenelse{\boolean{showeditcomments}}
            {\marginpar{\footnotesize{\textcolor{red}{[\textbf{#1:} \emph{#2}]}}}}
            {}
          }

\usepackage{latexsym,array,delarray,amssymb,amsmath}
\usepackage{graphicx,verbatim}
\usepackage{url,hyperref}

\newtheorem{thm}{Theorem}
\newtheorem{thmstar}[thm]{Theorem*}
\newtheorem{lemma}[thm]{Lemma}
\newtheorem{prop}[thm]{Proposition}
\newtheorem{propstar}[thm]{Proposition*}
\newtheorem{cor}[thm]{Corollary}

\newtheorem{defn}[thm]{Definition}
\newtheorem{ex}[thm]{Example}

\newcommand{\pp}{\mathbb{P}}

\newcommand{\bfv}{\mathbf{v}}

\newcommand{\ind}{\mbox{$\perp \kern-5.5pt \perp$}}

\newcommand\Span{\operatorname{Span}}

\begin{document}

\title{Identifiability of 2-tree mixtures for group-based models}

\author{Elizabeth~S.~Allman, Sonja~Petrovi\'c, John~A.~Rhodes, and~Seth Sullivant
\thanks{ E.S.~Allman  and J.A.~Rhodes are with the Department
of Mathematics and Statistics, University of Alaska Fairbanks, Fairbanks, AK, 99775-6660.
E-mail: e.allman@alaska.edu, j.rhodes@alaska.edu}
\thanks{S.~Petrovi\'c is with the Department of Mathematics, Statistics, and Computer Science,
University of Illinois at Chicago, Chicago, IL, 60607-7045.
E-mail: petrovic@math.uic.edu}
\thanks{S.~Sullivant is with the Department of Mathematics,
North Carolina State University, Raleigh, NC, 27695.
E-mail: smsulli2@ncsu.edu}
\thanks{The authors thank the Statistical and Applied Mathematical Sciences Institute, for funding and hospitality, as this research was conducted under its Program on Algebraic Methods in Systems Biology and Statistics. We also thank the National Science Foundation for support through grant DMS-0714830 for ESA and JAR, and grant DMS-0840795 for SS. Finally, Erich Kaltofen of NCSU generously shared computer resources provided under NSF grant DMS-0532140.}
}

\maketitle

\begin{abstract}
Phylogenetic data arising on two possibly different tree topologies might be mixed through several 
biological mechanisms, including incomplete lineage sorting or horizontal gene transfer in the case 
of different topologies, or simply different substitution processes on characters in the case of the 
same topology. Recent work on a 2-state symmetric model of character change showed that for 4 taxa such a mixture 
model has non-identifiable parameters, and thus it is theoretically impossible to determine the two tree 
topologies from any amount of data under such circumstances. Here the question of identifiability is 
investigated for 2-tree mixtures of the 4-state group-based models, which are more relevant to DNA 
sequence data.  Using algebraic techniques, we show that the tree parameters are identifiable for the 
JC and K2P models.  We also prove that generic substitution parameters for the JC mixture models 
are identifiable, and for the K2P and K3P models obtain generic identifiability results for mixtures on 
the same tree.  This indicates that the full phylogenetic signal remains in such mixtures, and that the 
2-state symmetric result  is thus a misleading guide to the behavior of other models. 
\end{abstract}





\section{Introduction}
\label{sec:intro}

A basic question concerning any statistical model is whether a
probability distribution arising from the model uniquely determines
the parameters that produced it.  If so, the parameters are said to be
\emph{identifiable}. Indeed, parameter identifiability is necessary
for the consistency of inference.

In phylogenetics, it is especially important that the tree parameter
of a model be identifiable, so that evolutionary histories can be
consistently inferred. For basic models of character evolution along a
tree, in which all sites behave independently and identically,
identifiability of both the tree and continuous parameters is
long-established.  However, as phylogenetic models grow in complexity,
it becomes increasingly difficult to analyze the models thoroughly
enough to be certain this property is retained. Indeed, mixture models
of all sorts present difficulties, though positive results have been
obtained for models with a small number of classes evolving on the
same tree \cite{Allman2006}, and those with scaled
$\Gamma$-distributed rates \cite{AllmanAneRhodes07}. However, even for
the GTR+$\Gamma$+I model, which is currently the most commonly used in
DNA data analysis, it is yet to be proved that trees are identifiable.

Several recent works, including \cite{MosVig}, \cite{StefVig2007},
\cite{Matsen2007}, and \cite{Matsen2008}, considered 2-class mixture
models in which the two classes evolve along possibly different
topological trees. Such models could describe instances of horizontal
transfer of genetic material between taxa, or incomplete lineage
sorting in sequences composed of several concatenated genes.  In
particular, Matsen and Steel \cite{Matsen2007} showed that under the
binary symmetric model of Cavender-Farris-Neyman, a 2-class mixture on
a single 4-taxon tree can exactly `mimic' a single class model on a
different tree. Because of the small size of the state space in this
model, its group-based structure, and the small size of the tree,
explicit calculations were possible to fully analyze this
situation. However, one should be cautious about extrapolating from
this result to a pessimistic view about identifiability of similar
phylogenetic mixtures. The mixture of \cite{Matsen2007} is an
11-parameter model producing a probability distribution in a
7-dimensional space, so it is certainly overparameterized. While this
dimension count does not guarantee non-identifiability of the tree, it
does explain why it might likely arise.

By either passing to models with larger state spaces, such as 4-state
models appropriate to DNA, or by considering trees relating more taxa,
the joint distribution of states at the leaves of the tree will be
embedded in a larger dimensional space. Thus we might hope to avoid
overparameterization issues through either of these modifications. As
the analysis of real biological data typically involves both of these
changes, these are the types of mixture models it is most desirable to
understand.

\medskip

Here we consider 2-class mixtures analogous to those in the works
above, but for larger trees and/or state spaces.  

We continue to work
with group-based models, 
focusing primarily on those for DNA, 
so that we retain the powerful tool of the Fourier/Hadamard coordinate
transformation. 

We also make use of computational algebra software to
perform calculations well beyond what could be done `by hand.' Our
results on identifiability are generally quite positive, and although
these group-based models are still special cases, we believe they
provide a better guide to the behavior of more realistic models than
those of \cite{Matsen2007}.

\medskip

This paper is organized as follows. In Section \ref{sec:identProblem}
we introduce 2-tree mixture models and the identifiability problem in
the algebraic setting.  Background on group-based models is covered in
Section \ref{sec:groupBasedModels}, from basic definitions through
their presentation in terms of Fourier coordinates.

Section \ref{sec:id-trees} deals with identifiability of the tree
parameters for Jukes-Cantor and Kimura 2-parameter mixture models on
two trees.  The main result, that tree parameters in such mixtures on
at least 4 taxa are generically identifiable, is Theorem
\ref{thm:maintwotree} and its corollary.  Even with generic tree
identifiability proved for 2-tree mixtures, a natural question is
whether a single-class (unmixed) model can be distinguished from a
2-tree mixture. (This is not answered by the previous result, since
while a single-class model is a special case of a 2-class model, it is
non-generic.) We investigate this problem in Section
\ref{sec:mix-and-unmix}.

Finally, in Section \ref{sec:id-params}, we turn to identifiability of
the continuous parameters of these models, assuming the tree
parameters are known. One feature of a part of our analysis is the use
of computational algebra software to obtain some results \emph{with
  very high probability}. Although technically these remain
conjectures, lacking rigorous proof, the conclusions we draw from such
calculations are highly reliable for theoretical reasons.  While using
calculations this way is familiar to applied algebraic geometers, this
approach may be new to others, so we begin the section by explaining
the reasoning informally. With this qualification, we establish the generic
identifiability of continuous parameters
for the JC model when either $n\ge5$, or $n=4$ and the trees are
distinct.  In the case of identical trees with $n \ge 5$ taxa, we
give a fully rigorous argument for the three group-based models: JC,
K2P, and K3P. An interesting non-identifiable case arises from the
Jukes-Cantor mixtures on two identical 4-taxon trees.

Command files and instructions for verifying all our computations
using the software Singular \cite{GPS09} can be found at the
supplementary materials website for this paper \cite{TTwebsite}.  We
include both computations supporting our arguments, and those
producing our examples.

\medskip

We would, of course, prefer to push the work here beyond the
group-based models, to include those more routinely used in current
data analysis. It is possible, after all, that the group-based models
are special enough that identifiability results for them do not carry
over to more elaborate models. However, our current computational and
theoretical tools are not sufficient for us to address questions for
more general models.


\section{Preliminaries}
\label{sec:identProblem}

Consider a phylogenetic model of $k$-state character change on
$n$-taxon trees (e.g., for $k=4$, the Jukes-Cantor model). We assume the taxa labelling the leaves are identified with $[n]=\{1,2,,\dots,n\}$.
Then for
each leaf-labelled tree $T$, there is a \emph{parameterization map} $\psi_T$ giving
the joint distribution of states at the leaves of the tree $T$ as
functions of continuous parameters. With $S_T$ denoting the continuous
parameter space on $T$, which we assume is some full-dimensional
subset of $\mathbb R^m$,
$$\psi_T:S_T \to \Delta^{k^n-1},$$ where $\Delta^{k^n-1}\subset[0,1]^{k^n}$
is the probability simplex comprised of non-negative real vectors
summing to 1.

Given such a model, the associated 2-tree mixture model has the
following parameterization maps: For every pair of $n$-taxon trees
$T_1$ and $T_2$ on the same taxa,
let $S_{T_1,T_2}=S_{T_1}\times S_{T_2}\times [0,1] $ and $$\psi_{T_1,T_2}:S_{T_1,T_2} \to
\Delta^{k^n-1},$$ be defined by
$$\psi_{T_1,T_2}(s_1,s_2,\pi)=\pi \, \psi_{T_1}(s_1)+ (1-\pi)\psi_{T_2}(s_2).$$
Here $\pi$ is the \emph{mixing parameter}, giving the proportion of i.i.d.~sites
that evolve along tree $T_1$.

We will only consider \emph{algebraic models}, for which the maps
$\psi_{T}$, and hence $\psi_{T_1,T_2}$, are defined by polynomial
formulas. This is a small restriction, as many models (\emph{e.g.},
standard continuous-time models) which are not polynomial can be
embedded in ones that are (\emph{e.g.}, the general Markov
model). Algebraic models can be studied from the perspective of algebraic geometry \cite{Drton2008}, after extending $\psi_T$ and $\psi_{T_1,T_2}$ to complex
polynomial maps, with images in $\mathbb C^{k^n}$. We refer to $S_T$ and $S_{T_1,T_2}$ as  \emph{stochastic parameter spaces}, to distinguish them from the \emph{complex parameter spaces} of these extensions.

We denote by $V_{T}$
the algebraic variety which is the Zariski closure of the image of $\psi_T$ in
the complex projective space $\mathbb P^{k^n-1}$. (See \cite{CLOS,Harris} for background in algebraic geometry.) Then the closure of the
image of $\psi_{T_1,T_2}$ is a variety called
the \emph{join} of $V_{T_1}$ and $V_{T_2}$, denoted by $$V_{T_1} *V_{T_2}.$$ 
The join can be described
geometrically as the smallest variety containing all lines
intersecting both $V_{T_1}$ and $V_{T_2}$. In the case $T_1=T_2$ the
join is called the \emph{secant} variety of $V_{T_1}$.

We use $\mathcal M_T$, and $\mathcal M_{T_1}*\mathcal M_{T_2}$, to
denote the image of the parameterization maps when applied only to the
stochastic parameter spaces $S_T$ and $S_{T_1,T_2}$. Thus these denote
the sets of all probability distributions arising from the
parameterized models, and $$\mathcal M_{T}\subsetneq V_{T},\ \
\mathcal M_{T_1}*\mathcal M_{T_2}\subsetneq V_{T_2}*V_{T_2}.$$ While
$\mathcal M_{T}$ and $\mathcal M_{T_1}*\mathcal M_{T_2}$ are of course
the objects of primary interest to phylogenetic applications, the larger complex varieties $V_T$
and $V_{T_1}*V_{T_2}$ are more amenable to algebraic study.

Another parameterization of a dense subset of $V_{T_1}*V_{T_2}$, which we will also use, is 
$$\phi_{T_1,T_2}:V_{T_1}\times V_{T_2}\times \mathbb P^1 \dashrightarrow V_{T_1}*V_{T_2},$$
which when restricted to an affine subset simply maps points on the two varieties to their convex sum using
the third coordinate as a weight.
(The dashed arrows indicates the map is only defined on a dense subset of the stated domain.)
If $\pi\in \mathbb C\subset\mathbb P^1$, then
\begin{equation}\psi_{T_1,T_2}(s_1,s_2,\pi)=\phi_{T_1,T_2}(\psi_{T_1}(s_1),\psi_{T_2}(s_2),\pi).\label{eq:compparam}
\end{equation}

\smallskip

Associated to any algebraic variety $V$ is the ideal $\mathcal{I} = \mathcal{I}(V)$ of polynomials that defines it; namely, a polynomial $f \in \mathcal{I}$ if, and only if, for any point $\bfv \in V$, $f(\bfv) = 0$.  For a variety associated to a phylogenetic model, such polynomials give constraints that entries of a distribution of states at the leaves of a tree must satisfy if it arises from the given model. First introduced in phylogenetics by Cavender and Felsenstein \cite{CF87} and Lake \cite{Lake87},  these polynomials are known as \emph{phylogenetic invariants}, and have been studied extensively in many papers, including \cite{Evans1993, SzEdStPe93, SzStEr93, Steel1995, Hendy1996,  Sturmfels2005, ARgm, ARmb}.

\medskip

For algebraic models, it is convenient to slightly weaken the notion of identifiability to \emph{generic identifiability}. The word `generic' is used to mean `except on a proper algebraic subvariety' of the parameter space.  Although it is sometimes possible to be explicit about this subvariety, we usually are not, since the key point in interpretation is that the subvariety is a closed set of Lebesgue measure 0 inside the larger set.  Thus regardless of the precise subvariety involved, `randomly' chosen points are generic with probability 1.
 
An additional issue for identifiability of 2-tree mixtures is class swapping: Interchanging the trees, along with their parameters, while replacing the mixing parameter $\pi$ by $1-\pi$, has no effect on the resulting distribution. Thus, a useful notion of identifiability must allow for this.

\begin{defn}
  The tree parameters of the 2-tree mixture model are 
  \emph{generically identifiable} if, for any binary trees  $T_1,T_2$ on the same set of taxa, and generic
  choices of $s_1, s_2, \pi$,
$$\psi_{T_1,T_2}(s_1,s_2,\pi) = \psi_{T'_1,T'_2}(s_1',s_2',\pi')$$ implies $\{T_1,T_2\}=\{T_1',T_2'\}.$
\end{defn}

\begin{defn}
  The continuous parameters of a $2$-tree mixture model on $T_1$ and $T_2$ are \emph{generically identifiable} if for generic choices of
  $s_1,s_2, \pi$, $$\psi_{T_1,T_2}(s_1,s_2,\pi) =  \psi_{T_1,T_2}(s_1',s_2',\pi')$$ implies
  $(s_1,s_2,\pi)=(s_1',s_2',\pi')$, or, in the case where $T_1 = T_2$,
  $(s_1,s_2,\pi)=(s_2',s_1',1-\pi)$.
\end{defn}

Let $K \subset [n]$ be a subset of the leaf set.  For any tree $T$ on
$n$ leaves, $T|_K$ will denote the induced subtree of $T$ with leaf set $K$.  Since
marginalization onto leaf subsets is a linear map that preserves the
mixture structure of a phylogenetic model, we obtain the following useful fact.

\begin{lemma}\label{lem:subtree}
  Let $T_1, T_2, T_3, T_4$ be $n$-taxon trees, not necessarily
  distinct, and let $K \subseteq [n]$.  If $V_{T_1|_K} \ast V_{T_2|_K}
  \not\subseteq V_{T_3|_K} \ast V_{T_4|_K}$, then $V_{T_1} \ast
  V_{T_2} \not\subseteq V_{T_3} \ast V_{T_4}$.
\end{lemma}

\begin{proof}
  Marginalization to a fixed set $K$ gives a linear map from $\mathbb{C}^{k^n}$ to
  $\mathbb{C}^{k^{|K|}}$, which sends $V_{T}$ to $V_{T|_K}$
  for any $T$.  For any linear transformation $A$ and any varieties,
  $V, W$, we have $A( V \ast W) = AV \ast AW$, because the mixture
  construction is a linear operation.  Since for any sets $S_1, S_2$, and any
  map $f$, $f(S_1) \not\subseteq f(S_2)$ implies that $S_1 \not\subseteq S_2$, the lemma follows.
\end{proof}


\section{Group-based Phylogenetic Models}
\label{sec:groupBasedModels}

Group-based phylogenetic models will be the main subject of study in this paper, so we collect known
results about these models, including their natural representation in Fourier coordinates.  

\smallskip

Throughout we assume that all trees $T$ are binary.  We root $T$ by picking an arbitrary  edge, and introducing a root $\rho$ as a distinguished node of degree $2$ along it.  Thus, every edge of $T$ may be considered directed away from the root.  To each node $v$ in the tree, we associate a discrete random variable $X_v$ with $k$ states, and to each directed edge $e$ in the tree, we associate a Markov transition matrix $M_e$, describing the conditional probabilities of various state changes.  We assume the reader is familiar with the usual assumptions of Markov processes on trees \cite{SemSt}.  The \emph{joint  distribution of states at the leaves of $T$} may be computed (the image of $\psi_T$), once the root distribution and the collection of $\{M_e\}$ are specified.  We refer to the entries of the root distribution and the Markov matrices as the \emph{continuous  parameters} of the model.

\begin{defn} 
Let  $T$ be a binary tree rooted at $\rho$. Let $G$ be a finite abelian group of order $k$, and identify its elements with the state space of the random variables $X_v$.  Then a \emph{group-based model} on $T$ for the group $G$ is a phylogenetic model with a uniform root distribution,  and transition probabilities on each edge $e$ satisfying $$M_e(g,h)  = f_e(g - h)$$ for some functions $f_e : G \to \mathbb{R}$.
\end{defn}

Some standard examples of group-based models are the binary symmetric model, a.k.a.~the Cavender-Farris-Neyman (CFN) model, which is associated to the group $\mathbb Z_2$; and the Jukes-Cantor (JC) model, the Kimura 2-parameter (K2P) model, and the Kimura 3-parameter (K3P) model, which are associated to the group $\mathbb Z_2\times \mathbb Z_2$. With appropriate ordering of the state spaces, the transition matrices for these models have the following forms, respectively:
$$
\begin{pmatrix}
\alpha & \beta \\
\beta & \alpha
\end{pmatrix}, \quad 
\begin{pmatrix}
\alpha & \beta & \beta & \beta \\
\beta & \alpha & \beta & \beta \\
\beta & \beta & \alpha & \beta \\
\beta & \beta & \beta & \alpha
\end{pmatrix} ,
$$
$$
\begin{pmatrix}
\alpha & \beta & \gamma & \gamma \\
\beta & \alpha & \gamma & \gamma \\
\gamma & \gamma & \alpha & \beta \\
\gamma & \gamma & \beta & \alpha
\end{pmatrix}, 
\quad
\begin{pmatrix}
\alpha & \beta & \gamma & \delta \\
\beta & \alpha & \delta & \gamma \\
\gamma & \delta & \alpha & \beta \\
\delta & \gamma & \beta & \alpha
\end{pmatrix}.
$$

Group-based models are subject to a remarkable linear change of coordinates, called the discrete Fourier, or Hadamard, transform \cite{Hendy89, Hendy1993, Evans1993, SzEdStPe93, SzStEr93}.  After applying the Fourier transform the models are seen to be toric varieties \cite{Sturmfels2005}.  In particular, the transformed image coordinates are given in terms of transformed domain coordinates by a monomial parameterization.  

To make this parameterization explicit, henceforth let $G$ be $\mathbb Z_2$ or $\mathbb Z_2 \times \mathbb Z_2$, and $T$ an $n$-taxon tree.  The Fourier coordinates for a group-based model are denoted $q_{g_1,  \ldots, g_n},$ where $g_i \in G$ for all $i$.  Let $\Sigma (T)$ be the set of splits induced by the edges of $T$.  To each split $A|B\in \Sigma(T)$, we associate a set of parameters: $a^{A|B}_g$ where $g \in G$. The toric parameterization for the model is then given by:

\begin{equation}\label{eq:fourierparam}
  q_{g_1, \ldots, g_n} = \left\{ \begin{array}{cc}
      \prod_{A|B \in \Sigma(T)}  a^{A|B}_{\sum_{i\in A} g_i}, &  \mbox{ if } \sum_{i =1}^n g_i  = 0\\
      0, &  \mbox{ otherwise}.
\end{array} \right.
\end{equation}

Note that by our choices of $G$, when $\sum_{i =1}^n g_i = 0$ we will have $\sum_{i \in A} g_i = \sum_{i \in B} g_i$ for
any split $A|B$. Thus the formula above does not depend on the ordering of the sets in the splits.

\smallskip

To ease notation, we describe trees by omitting trivial splits associated to leaf edges (\emph{i.e.}, those of the form $\{i\}\,|\, [n] \smallsetminus \{i\}$).  When describing the toric parametrization of a group-based model, we abbreviate the parameters associated with the edge leading to leaf $i$ by $a^{i}_{g}$.

For elements in the group $G = \mathbb{Z}_2 \times \mathbb{Z}_2$ associated to the Jukes-Cantor and Kimura models, we (arbitrarily) identify nucleotides with the group elements in the following way: $A = (0,0)$, $C = (0,1)$, $G = (1,0),$ and $T = (1,1)$.  We illustrate these notions with an example on a $5$-taxon tree.

\begin{ex}\label{ex:5leaffourier}
  Let $T = \{12|345, 123|45\}$.  The toric parameterization for the
  K3P model is given by formulas of the form:
$$
q_{g_1g_2g_3g_4g_5}  =  a^1_{g_1} a^2_{g_2} a^3_{g_3} a^4_{g_4} a^5_{g_5} a^{12|345}_{g_1 + g_2} a^{123|45}_{g_1 + g_2 + g_3}
$$
where $\sum_{i = 1}^5 g_i = 0$, and  $q_{g_1g_2g_3g_4g_5} = 0$ otherwise.  For example,
$$
q_{CCCTG} = a^1_{C} a^2_{C} a^3_{C} a^4_{T} a^5_{G} a^{12|345}_{A} a^{123|45}_{C}.
$$
\end{ex}

For the JC and K2P models, the Fourier coordinates are described by simply imposing
additional relationships on the continuous parameters.

\begin{prop}\cite{Evans1993,Hendy1993}
  In Fourier coordinates, the K3P model on a tree $T$ consists of
  all the Fourier vectors arising from representation
  (\ref{eq:fourierparam}) so that, for each split $e$, $a^e_A =1$,
  and 
  $a^e_C$, $a^e_G$, $a^e_T \in (0,1]$.  The K2P model is
  the submodel of the K3P model satisfying additionally, that for all splits
  $e$, $a^e_G = a^e_T$.  The JC model is the submodel of the
  K2P model satisfying additionally, that for all splits $e$, $a^e_C = a^e_G =
  a^e_T$.
\end{prop}
 
Significantly for the work later in this article, the Fourier transform is a \emph{linear} change of coordinates. Thus the
operation of taking tree mixtures commutes with the Fourier transform, which allows us to naturally represent the mixture models we consider in Fourier coordinates. Though these mixture models are not toric, we still gain insight from this viewpoint.

\smallskip

To close this section, we make several comments regarding some combinatorial aspects of Fourier coordinates for group-based models. As linear invariants are crucial to the arguments below, we first discuss enumeration of distinct Fourier coordinates, and computations of the dimension of the space of linear invariants for a model.  In closing, we illustrate some useful combinatorial mnemonics for working with Fourier coordinates and identifying invariants. Although these devices are likely familiar to experts in group-based models, we hope our exposition will be useful to those less familiar with these models.

The zero set of the linear invariants for any variety $V\subset \mathbb P^n$ is the smallest linear
subspace of $\pp^n$ containing $V$. This set is called
 the \emph{span of $V$}, $\Span(V)$.  The span of a finite collection of varieties is
defined similarly, as the span of their union. Note that, by the join construction, it is
immediate that $\Span(V \ast W) = \Span(V, W)$.

\smallskip

For group-based models on an $n$-taxon tree (that is, an \emph{unmixed} model), the number of distinct Fourier coordinates is precisely the
dimension of the span of the model, as there are no linear relations between distinct monomials. This establishes:
\begin{prop}
For the CFN and K3P models, there are no non-trivial linear
invariants.  The number of distinct Fourier coordinates is $2^{n-1}$
for the CFN model and $4^{n-1}$ for the K3P model.
\end{prop}

As our method of investigation of mixture models in Section \ref{sec:id-trees} depends upon the existence of linear invariants, we thus focus on the JC and K2P group-based models.

Steel and Fu \cite{Steel1995} computed the dimension of
$\Span(V_T)$ for the JC model.  Hendy and Penny \cite{Hendy1996} performed a  similar
computation for the K2P model.

\begin{thm}\label{thm:lininvdim} \cite{Steel1995,Hendy1996} Let $T$ be an $n$-taxon binary tree.
Then, for the JC model on $T$, the number of distinct Fourier coordinates is the Fibonacci number $F_{2n-2}$, satisfying the recurrence
$$F_0 = 1, \, F_1 = 1, \, F_{i} = F_{i-1} + F_{i-2}.$$
That is, for the JC model $\Span(V_T)$ has dimension $F_{2n-2}$.

For the K2P model on $T$, the
number of distinct Fourier coordinates is $H_n$, satisfying the
recurrence:
$$ H_1 =1, \, H_2 = 3, \, H_i = 4H_{i-1} - 2H_{i-2}. $$
That is, for the K2P model $\Span(V_T)$ has dimension $H_n$. 
\end{thm}

Fourier coordinates for group-based models have useful
combinatorial representatives in terms of labelled or colored versions
of the underlying tree $T$.  For this representation, we associate a
color to each of the different parameter classes in the model. For
example, in the JC model, there are two parameter classes: the $A$
class (grey), and the \{$C,G,T$\} class (black).  With this choice of
colors, in the parametric description of $q_{g_1, \ldots, g_n}$ if a
parameter $a^{B|B'}_A$ occurs, then we color the edge corresponding to
the split $B|B'$ grey.  If a parameter $a^{B|B'}_C, a^{B|B'}_G,$ or
$a^{B|B'}_T$ occurs, we color the edge corresponding to the split $B|B'$
black. As shown in \cite{Steel1995}, this establishes a correspondence between distinct Fourier 
coefficients for the JC model and subforests of $T$ with the same leaf set $[n]$.

The color-coding of the underlying tree works similarly for the K2P model; here we have three classes of parameters, the $A$-class, the $C$-class, and the $\{G,T\}$-class, and hence use three colors.

\begin{ex}
  Continuing Example \ref{ex:5leaffourier}, the colored diagram
  corresponding to the Fourier coordinate $q_{CCCTG}$ for the JC model
  is shown in Figure \ref{fg:fg5-jc-bw.eps}.
\begin{figure}[h]
\begin{center}
\includegraphics[height=1in]{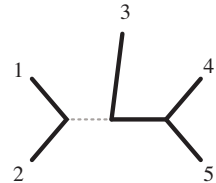}
\caption{JC-coloring for $q_{CCCTG}$. (Key: $A$-class = dashed grey, \{$C,G,T$\}-class = black)}\label{fg:fg5-jc-bw.eps}
\end{center}
\end{figure}

For the K2P model, the same Fourier coordinate is 
represented by the tri-colored tree of Figure \ref{fg:fg5-k2p-bw.eps}. 
\begin{figure}[h]
\begin{center}
\includegraphics[height=1in]{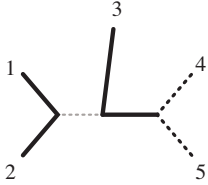}
\caption{K2P-coloring for $q_{CCCTG}$.  (Key: $A$-class = 
dashed gray, $C$-class = black, $\{G,T\}$-class
 = dashed black)}\label{fg:fg5-k2p-bw.eps}
\end{center}
\end{figure}

These diagrams are useful for determining the invariants 
that a particular group-based model satisfies.  For instance, 
one phylogenetic invariant for the JC model on the tree $T$ with split 
$12|34$ is given in Fourier coordinates by
\begin{equation}q_{CTGA}q_{ACTG} = q_{CGCG}q_{ACCA}.\label{eq:4invJC}\end{equation}
This  relationship may be represented in pictorial form by the diagram
in Figure \ref{fg:fg-jc-invDiagram.eps}.

\begin{figure}[h]
\begin{center}
\includegraphics[height=1.5in]{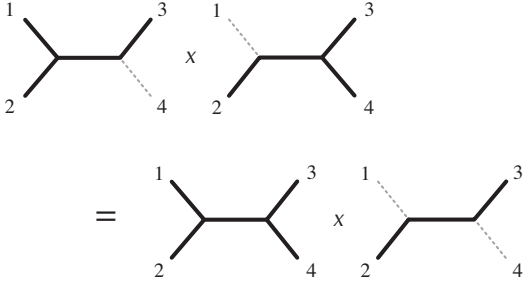}
\caption{Pictorial view of invariant \eqref{eq:4invJC} for the JC model on $T$.}\label{fg:fg-jc-invDiagram.eps}\end{center}
\end{figure}
\end{ex}


\section{Identifiability of Tree Parameters}
\label{sec:id-trees}

The goal of this section is to prove that the tree parameters are
generically identifiable for $2$-tree JC and K2P mixture models on at least 4 taxa.  For the complex varieties associated to the models
this is formalized as follows:

\begin{thm}\label{thm:maintwotree} 
  Suppose $T_1, T_2, T_3, T_4$ are binary trees, not necessarily
  distinct, on $n \geq 4$ taxa, and consider the $2$-tree mixture
  varieties for the JC and K2P models.  If $\{T_1, T_2\} \neq \{T_3,
  T_4\}$, then $V_{T_1} \ast V_{T_2} \not\subseteq V_{T_3} \ast
  V_{T_4}$.  
\end{thm}

If $\{T_1, T_2\} \neq \{T_3,
  T_4\}$, then the noncontainment of the irreducible varieties $V_{T_1} \ast V_{T_2}$ and $V_{T_3} \ast V_{T_4}$ in one another shows their intersection is a proper subvariety of strictly lower dimension. The preimage of this intersection under the complex parameterization map is thus a proper subvariety of the parameter space.   Since the subset of stochastic parameters is Zariski dense in the complex parameter space,
those stochastic parameters mapping to the intersection also lie in a closed set of Lebesgue measure 0.  Thus we obtain the main result of the section:
  
\begin{cor} For the $2$-tree JC and K2P mixture models on at least 4 taxa the tree parameters
are generically identifiable for either stochastic or complex parameters. 
\end{cor}

The proof of Theorem \ref{thm:maintwotree} proceeds in three parts.  First, we show that the two tree
parameters when $T_1 = T_2$ are identifiable.  Then we focus on the
quartet trees, constructing a linear invariant that completes the proof of the result
for quartet tree mixtures.  Finally, we combine our quartet results,
the six-to-infinity theorem of Matsen, Mossel, and Steel
\cite{Matsen2008}, and a linear invariant for $6$-taxon tree mixtures, to
deduce identifiability of trees for $2$-tree mixtures on an arbitrary number of
leaves.  


\subsection{$2$-Tree Mixture with $T_1 = T_2$}

In this section, we focus on a mixture of a tree with itself; that is,
we study the secant variety $V_T \ast V_T$.  We show that $V_T \ast
V_T$ can be distinguished from any $2$-tree mixture variety $V_{T_1}
\ast V_{T_2}$, provided $T_1$ and $T_2$ are not both $T$.

\begin{prop}\label{prop:sametree}
  Let $T_1, T_3, T_4$ be three binary trees, not necessarily
  distinct, with $n \geq 4$ leaves, such that $\{T_1\} \neq \{T_3,
  T_4\}$.  Then under both the JC and K2P $2$-tree mixture
  models, $V_{T_1} \ast
  V_{T_1} \not\subseteq V_{T_3} \ast V_{T_4}$ and $V_{T_3} \ast
  V_{T_4} \not\subseteq V_{T_1} \ast V_{T_1}$.
\end{prop}

\begin{proof}  Assume $T_3\ne T_1$. By \cite{Steel1995} and \cite{Hendy1996}, for both the
  unmixed JC and K2P models, 
  there exists a linear invariant $l \in \mathcal{I}(V_{T_1})
  \setminus \mathcal{I}(V_{T_3})$.  Since the set of linear invariants of $V_{T_1}$ and
  $V_{T_1} \ast V_{T_1}$ coincide (the varieties have the same span),
there exists a linear invariant $l \in
  \mathcal{I}(V_{T_1} \ast V_{T_1}) \setminus \mathcal{I}(V_{T_3})$.
  Hence, $V_{T_3} \not\subseteq V_{T_1} \ast V_{T_1}$.  Now since
  $V_{T_3} \subseteq V_{T_3} \ast V_{T_4}$, it follows that $V_{T_3} \ast
  V_{T_4} \not\subseteq V_{T_1} \ast V_{T_1}$.

It remains to show $V_{T_1} \ast V_{T_1} \not\subseteq V_{T_3} \ast V_{T_4}$. 
In fact, it is enough to show that
\begin{equation}\dim V_{T_3} \ast V_{T_4} \leq \dim V_{T_1} \ast V_{T_1}.\label{eq:dimineq}\end{equation} Indeed,  if
 this inequality holds strictly, the claim is obvious. If, on the other hand, the dimensions are equal, then
 since both of the joins are irreducible varieties, $V_{T_1} \ast V_{T_1} \subseteq V_{T_3} \ast V_{T_4}$  would imply equality of varieties, contradicting the anti-containment already established.
  
Now a simple bound on the dimension of a join, coming from its natural parameterization, is
$$\dim V \ast W \leq  \dim V + \dim W + 1,$$ where the quantity on the right is called the \emph{expected dimension}  when it is no larger than
the dimension of the ambient space.
 In the case of the JC model, $\dim(V_{T_3})=\dim(V_{T_4})=2n-3$,  which shows $\dim(V_{T_3}\ast V_{T_4})\le 4n-5$. 
Similarly, $ \dim (V_{T_3} \ast V_{T_4}) \leq 8n-11$ for the K2P model.

To complete the proof of the claim of Proposition \ref{prop:sametree} for the JC model by establishing inequality \eqref{eq:dimineq}, it suffices to show the secant variety has the expected dimension, as
given by the following:

\begin{lemma}\label{prop:sametreedim} 
If $T$ is an $n$-taxon binary tree then
$\dim V_T \ast V_T = 4n-5$ for the JC model. 
\end{lemma}

As the proof of this lemma is more involved, we defer it until after our current argument.
For the K2P model, we prove below a weaker claim.

\begin{lemma}\label{prop:sametreeK2}
  If $T$ is a 4-taxon binary tree, then
  $\dim V_T \ast V_T
  = 21$ for the K2P model.
\end{lemma}

This is sufficient to complete the proof  of the claim of Proposition \ref{prop:sametree} in the K2P case for 4-taxon trees. Larger trees are then treated by considering marginalizations to induced 4-taxon trees: Choose a set $K$ of 4 taxa for which the induced quartet trees  $T_1|_K,T_3|_K,T_4|_K$ are not all the same, and then apply Lemma \ref{lem:subtree}.
\end{proof}

\smallskip

To prove Lemma \ref{prop:sametreedim}, we make
use of a special case of the tropical secant varieties theory of
Draisma \cite{Draisma2008} and the fact that the varieties $V_T$ are
toric varieties.  To explain Draisma's result (Theorem \ref{thm:draisma} below), we introduce some background material on toric varieties and convex geometry.

Recall that a toric variety is specified as the image of a polynomial map, each of whose coordinate functions is a monomial.  As a monomial is of the form $x^u = x_1^{u_1} x_2^{u_2} \cdots x_d^{u_d}$, we associate to each monomial a non-negative integer vector $u$.  To a toric variety, we associate a collection of non-negative integer vectors $A \subset \mathbb{N}^d$, one vector for each monomial appearing in the parameterization.  We also identify $A$ with a matrix whose columns are the given set. The toric variety is often denoted $V_A$.  Algebraic and geometrical properties of toric varieties are reflected in corresponding properties of the vector configuration $A$ \cite{Fulton1993,Sturmfels1996}.

By a hyperplane in $\mathbb{R}^d$, we mean a linear hypersurface $H = \{ x \in \mathbb{R}^d :  c^Tx = e \}$.  The complement $\mathbb{R}^d \setminus H$ consists of two connected components, which we denote by $H^+$ and $H^-$.  

\begin{thm}[\cite{Draisma2008}] \label{thm:draisma}
  Let $V_A$ be a projective toric variety, with corresponding set of
  exponent vectors $A \subseteq \mathbb{N}^d$.  Suppose that $A$ has
  rank $r$, so that $\dim V_A = r-1$.  Let $H$ be a hyperplane  not intersecting $A$.  Let $A^+ = A
  \cap H^+$ and $A^- = A \cap H^-$.  Then $\dim V_A \ast V_A \geq {\rm
    rank} \, A^+ + {\rm rank} \, A^- - 1$.  In particular, if there
  exists an $H$ such that ${\rm rank} \, A^+ = {\rm rank} \, A^- =
  {\rm rank} \, A$, then $V_A$ has the expected dimension.
\end{thm}

\begin{proof}[Proof of Lemma \ref{prop:sametreedim}]
To apply Theorem \ref{thm:draisma}, we must investigate the vector configurations $A$ associated with the Jukes-Cantor model and find a hyperplane with the desired properties.  For the Jukes-Cantor model on a binary tree, each distinct Fourier coordinate corresponds to a subforest $F$ of the tree, and the corresponding monomial has the form
$$
\prod_{e \in F}  a^e_{C}  \prod_{e \notin F} a^e_A.
$$
(Here we consider the $a^e_A$ as variables, rather than setting them to be 1. This simply homogenizes our parameterization.)
The vector corresponding to $F$ is in $\mathbb{N}^{4n -6}$; specifically, $u^F = (x_e, y_e)_{e \in \Sigma(T)}$ such that $x_e = 1$ and $y_e = 0$ if $e \in F$, and $x_e = 0$ and $y_e = 1$ if $e \notin F$.  For example, in the case that $n = 4$, and $T$ is the tree with nontrivial split $12|34$, then, after removing repeated columns, $A$ consists of the columns of the $10 \times 13$ matrix:

{\small $$
\left( \begin{array}{ccccccccccccc}
0 & 1 & 1 & 1 & 0 & 0 & 0 & 1 & 1 & 1 & 0 & 1 & 1 \\
0 & 1 & 0 & 0 & 1 & 1 & 0 & 1 & 1 & 0 & 1 & 1 & 1 \\
0 & 0 & 1 & 0 & 1 & 0 & 1 & 1 & 0 & 1 & 1 & 1 & 1 \\
0 & 0 & 0 & 1 & 0 & 1 & 1 & 0 & 1 & 1 & 1 & 1 & 1 \\
0 & 0 & 1 & 1 & 1 & 1 & 0 & 1 & 1 & 1 & 1 & 0 & 1 \\
1 & 0 & 0 & 0 & 1 & 1 & 1 & 0 & 0 & 0 & 1 & 0 & 0 \\
1 & 0 & 1 & 1 & 0 & 0 & 1 & 0 & 0 & 1 & 0 & 0 & 0 \\
1 & 1 & 0 & 1 & 0 & 1 & 0 & 0 & 1 & 0 & 0 & 0 & 0 \\
1 & 1 & 1 & 0 & 1 & 0 & 0 & 1 & 0 & 0 & 0 & 0 & 0 \\
1 & 1 & 0 & 0 & 0 & 0 & 1 & 0 & 0 & 0 & 0 & 1 & 0 
\end{array} 
\right).
$$}The first 2 rows here correspond to the $x_e$ for edges in one cherry on the tree, the next 2 to $x_e$ for edges in the other cherry, and the 5th to $x_e$ for the central edge; the last 5 correspond to the $y_e$, with edges in the same order. Thus, the first column corresponds to the empty forest, the second to the first cherry, and so on.

Consider the hyperplane  $H = \{(x_e,y_e) \in \mathbb{R}^{4n-6} : \sum_{e \in \Sigma(T)} x_e = |\Sigma(T)|
  -3/2$.
  This means that the vectors on one side of the partition will
  correspond to subforests of $T$ having at most one edge of $T$
  missing.  The subforests on the other side will have at least two
  edges missing.  Call the first set of vectors $A^+$ and the second set
  of vectors $A^-$.  In the matrix above, $A^+$ consists of the last six columns and $A^-$ consists of the first seven columns.

  The first set $A^+$ contains exactly $|\Sigma(T)| + 1$ vectors, since the tree
  itself is a subforest and removing any edge always produces a
  subforest.  This set thus forms the vertices of a simplex of
  dimension equal to the number of edges, so $A^+$ has rank $2n-2$.  
  
  The second set, $A^-$, contains the empty graph and all
  paths between pairs of vertices.  If we restrict attention to only those vectors corresponding to the paths between pairs of vertices, this gives us the exponent vectors of the toric varieties corresponding to toric degenerations of the Grassmannian \cite{Speyer2004}, which has rank $2n-3$.  Adding the vector corresponding to the empty subforest increases the rank by one.
\end{proof}

\smallskip

\begin{proof}[Proof of Lemma \ref{prop:sametreeK2}]
To apply Theorem \ref{thm:draisma}, we must investigate the vector configurations $A$ associated with the K2P model and find a hyperplane with the desired properties.
For $n = 4$, there are $H_4 = 34$ distinct Fourier coordinates.  We focus on the tree with the unique nontrivial split $12|34$.  Each monomial in the Fourier parameterization has the form
$$
a^1_{g_1} a^2_{g_2} a^3_{g_3} a^4_{g_4} a^{12|34}_{g_1 + g_2},
$$
where $a^e_{G} = a^e_T$ for all splits $e$.  This implies that the matrix $A$ is a $15 \times 34$ matrix.   The coordinates on $\mathbb{R}^{15}$ are $x_e, y_e, z_e$, where for a given edge $e$,  
$$(x_e, y_e, z_e) = \left\{  \begin{array}{cl}
(1,0,0) & \mbox{ if } g_e = A \\
(0,1,0) & \mbox{ if } g_e = C \\
(0,0,1) & \mbox{ if } g_e = G,T.
\end{array} \right.
$$
Since the model has dimension $10$, we see that the matrix $A$ has rank $11$. 

Now consider the hyperplane $H = \{ (x_e, y_e, z_e) \in \mathbb{R}^{15} : \sum_{e \in \Sigma(T)} y_e + z_e  = 7/2 \}$.  A direct calculation shows that this partitions $A$ into $A^+$ and $A^-$ each with rank $11$, and completes the proof.
\end{proof}


\subsection{Linear Invariants for Quartet Mixtures}

We next focus on quartet trees.  The three
fully-resolved quartet trees will be indicated by their non-trivial splits: $T_{12|34}$, $T_{13|24}$, and
$T_{14|23}$.  The main result of this section is that linear invariants can generically identify 2-quartet mixtures.

\begin{lemma}\label{lem:linearinv}
  For both the JC and K2P models, the linear polynomial
$$f  =  q_{GGGG} + q_{GTGT} - q_{GGTT} - q_{GTTG}$$
satisfies $f \in \mathcal{I}(V_{T_{12|34}} \ast V_{T_{14|23}})
\setminus \mathcal{I}(V_{T_{13|24}}).$

\end{lemma}
Note the lemma further implies
$$f \in \mathcal{I}(V_{T_{12|34}} \ast V_{T_{14|23}})
\setminus \mathcal{I}(V_{T_{13|24}}\ast V_{T_{14|23}}).$$
Combining this with Proposition
\ref{prop:sametree}, we deduce a first case of Theorem \ref{thm:maintwotree}:

\begin{cor}\label{cor:quartets} The case $n=4$ of 
Theorem \ref{thm:maintwotree} holds.\end{cor}

\begin{proof}[Proof of Lemma \ref{lem:linearinv}]
 Denote the parameters for tree $T_{12|34}$ by $a$
  and the parameters for $T_{14|23}$ by $b$.  We must show that $f =
  0$ whenever we substitute for the $q$'s the parameterization for the  mixture model.
One checks that:
$$
q_{GGGG}  =  \pi  a_G^1 a_G^2 a_G^3 a_G^4 a_A^{12|34} + (1-\pi) b_G^1 b_G^2 b_G^3 b_G^4 b_A^{14|23}, 
$$
$$
q_{GTGT}  =  \pi  a_G^1 a_T^2 a_G^3 a_T^4 a_C^{12|34} + (1-\pi) b_G^1 b_T^2 b_G^3 b_T^4 b_C^{14|23},
$$
$$
q_{GGTT}  =  \pi  a_G^1 a_G^2 a_T^3 a_T^4 a_A^{12|34} + (1-\pi) b_G^1 b_G^2 b_T^3 b_T^4 b_C^{14|23},
$$
$$
q_{GTTG}  =  \pi  a_G^1 a_T^2 a_T^3 a_G^4 a_C^{12|34} + (1-\pi) b_G^1 b_T^2 b_T^3 b_G^4 b_A^{14|23}. 
$$
Since for the K2P and JC models $a^e_G = a^e_T$, $b^e_G = b^e_T$ for all $e$, these formulae show $f = 0$, as can be checked using
color-codes trees such as in Section \ref{sec:groupBasedModels}.
 Thus  $f \in \mathcal{I}(V_{T_{12|34}} \ast
  V_{T_{14|23}})$.  

On the other hand, for the tree $T_{13|24}$ we have:
$$
q_{GGGG}  =    c_G^1 c_G^2 c_G^3 c_G^4 c_A^{13|24}
$$
$$
q_{GTGT}  =    c_G^1 c_T^2 c_G^3 c_T^4 c_A^{13|24}
$$
$$
q_{GGTT}  =    c_G^1 c_G^2 c_T^3 c_T^4 c_C^{13|24}
$$
$$
q_{GTTG}  =    c_G^1 c_T^2 c_T^3 c_G^4 c_C^{13|24}
$$
Even though in the JC model $c^e_C=c^e_G = c^e_T$  and $c^e_A=1$ for all $e$, $f$ is not identically zero when evaluated at these expressions. Thus
$f\notin \mathcal{I}(V_{T_{13|24}})$ for the JC model, and hence also for the K2P model.
\end{proof}


\subsection{From Quartets to Sextets and Beyond}

Identifiability of quartet mixtures can be used to show
identifiability for larger trees by marginalization of tree
models and their mixtures.  However, it is not, in general, possible
to identify two trees from the union of their sets of induced quartet trees. Thus this approach requires some care.
That all difficulties arise from trees of at most 6 taxa is the content of the following
combinatorial theorem of Matsen, Mossel, and Steel. 

\begin{thm}[Six-to-Infinity Theorem] \cite{Matsen2008}
\label{thm:sixToInfty}
  Suppose that the tree parameters $T_1, T_2$ are identifiable for a
  2-tree phylogenetic mixture model for binary trees with six leaves.  Then tree
  parameters are identifiable for binary trees with $\geq 6$
  leaves.
\end{thm}

Combining the results of Corollary \ref{cor:quartets} and Lemma
\ref{lem:subtree}, we have that $V_{T_1} \ast V_{T_2} \not\subseteq
V_{T_3} \ast V_{T_4}$, if there is a four element subset $Q \subseteq
[n]$ such that $\{T_1|_Q, T_2|_Q \} \neq \{T_3|_Q, T_4|_Q \}$.

It remains to show that $V_{T_1} \ast V_{T_2} \not\subseteq V_{T_3}
\ast V_{T_4}$ for pairs of trees such that $\{T_1|_Q, T_2|_Q \} =
\{T_3|_Q, T_4|_Q \}$ for \emph{all} four element subsets $Q \subseteq [n]$.
Let $\mathcal Q(T_i,T_j)$ denote the multiset of all quartet trees
$T_i|_Q$, $T_j|_Q$ induced by $T_i$ and $T_j$. We say two pairs of
trees $T_1, T_2$ and $T_3, T_4$ are \emph{quartet-matched} if
$\mathcal Q(T_1,T_2)=\mathcal Q(T_3,T_4)$.

\begin{prop}\label{prop:quartetmatched} 
  For $n = 5$ leaves, any two quartet-matched pairs of trees $T_1, T_2$
  and $T_3, T_4$ has $\{T_1, T_2\} = \{T_3, T_4\}$.  For $n = 6$
  leaves, every quartet-matched pair of trees with $\{T_1, T_2\} \neq
  \{T_3, T_4\}$ is equivalent, up to $\mathfrak S_6$ symmetry, to the pairs
  defined by
  \begin{align} \label{eq:2-tree-pairs}
    T_1 &= \{ 12|3456, 123|456, 1234|56\},\notag\\ \quad T_2 &= \{13|2456, 123|456, 1235|46 \}, \\
    T_3 &= \{ 13|2456, 123|456, 1234|56\},\notag\\ \quad T_4 &=
    \{12|3456, 123|456, 1235|46 \}\notag.
\end{align}
\end{prop}

\begin{proof} Fix two binary trees $T_1,T_2$ with $n$ leaves.
If the trees are identical, the result is clear, so we assume throughout $T_1\ne T_2$.

\medskip

Consider first $n=5$.

If leaves $j,k$ form a cherry in $T_i$, then they will also form a
cherry in all 3 quartet trees including $j,k$ induced from $T_i$. On
the other hand, if they do not form a cherry in $T_i$, then they will
form a cherry in either 0 or 1 of these 3 induced quartet trees.  Thus
by counting the elements of the multiset $\mathcal Q(T_1,T_2)$ with
each possible cherry $j,k$, we can determine which cherries occur in
both trees (count 6), which occur in exactly one tree (count 3 or 4),
and which occur in no trees (count 0, 1, or 2).

If a cherry occurs in both trees, suppose it is $\{1,2\}$. Then from
considering the quartets on $\{2,3,4,5\}$ both $T_1$ and $T_2$ are
determined.

If the two trees have no cherry in common, then we know the 4 distinct
cherries that occur in the 2 trees. If only 4 taxa occur in these 4
cherries, then we may uniquely pair them according to their
compatibility, and the two 5-taxon trees $T_1$ and $T_2$ are
determined.  If all 5 taxa occur in these cherries, since the cherries
are distinct we may assume they are $\{1,2\}$ and $\{3,4\}$ (from one
tree), and $\{1,5\}$ and $\{2,3\}$ (from the other), though we
initially do not know which come from which tree. However, we again
see that these can be uniquely paired for compatibility, and thus
$T_1$ and $T_2$ are determined.

\smallskip

Now consider $n=6$.

By the $n=5$ case, we may determine the multiset $\mathcal F=\mathcal
F(T_1,T_2)$ of all 5-taxon induced trees from $T_1$ and $T_2$, so we
work with it instead of $\mathcal Q=\mathcal Q(T_1,T_2)$.

By counting cherries in $\mathcal F$, we may determine those possible
cherries that occur in both trees (count 8), exactly one tree (count 4
or 5), or no trees (count 0,1, or 2).
If a cherry occurs in both trees, suppose it is $\{1,2\}$. Then from
considering the 5-taxon trees on $\{2,3,4,5,6\}$ both $T_1$ and $T_2$
are determined.

For the reminder of the proof, we assume the trees have no cherry in common.
Thus either 4, 5, or 6 distinct cherries occur in $T_1$ and $T_2$. 
In the case of 6 distinct cherries, compatibility of cherries determines $T_1$ and $T_2$.

In the case of 5 distinct cherries, one of the trees must be
symmetric, and the other a caterpillar. Either compatibility of
cherries determines the symmetric tree (in which case both trees are
determined by removing the quartets from this tree from $\mathcal Q$
and using the remaining ones to construct the second tree), or we may
assume the 5 cherries have the form $\{1,2\}$, $\{3,4\}$, $\{5,6\}$
(from one tree) and $\{1,3\}$, $\{2,4\}$ (from the other), though of
course we do not know which come from which tree.  Since the cherry
$\{5,6\}$ is identified by this, consider the two elements of
$\mathcal F$ on $\{1,2,4,5,6\}$. As $\{5,6\}$ is a cherry in only one
of these trees, and that one also has $\{1,2\}$ as its other cherry,
this identifies $\{1,2\}$.  Thus $\{3,4\}$ is also known, and thus
one, and hence both, of the trees are determined.

In the case of 4 distinct cherries, both $T_1$ and $T_2$ are
caterpillars. We first investigate whether we can determine which
pairs of cherries occur on the same tree $T_i$. Since at least 2 of
the 4 cherries must be incompatible, let these be $\{1,2\}$ and
$\{1,3\}$. Either compatibility determines which other cherries these
are paired with, or the remaining cherries have the form $\{i,j\}$
with $i,j\in\{4,5,6\}$, and we may assume the cherries are $\{4,6\}$
and $\{5,6\}$.

If compatibility determined the cherries on $T_1$ as $\{1,2\}$ and
$\{i,j\}$, and those on $T_2$ as $\{1,3\}$ and $\{k,l\}$, then we may assume $j\ne 3$.
Then the two
elements of $\mathcal F$ on all taxa but $j$ can be matched with
the $T_i$ depending on whether they display the cherry $\{1,2\}$ or
$\{1,3\}$. This determines $T_1$, and hence $T_2$ as well.

This leaves only the case where the 4 cherries are $\{1,2\}$,
$\{1,3\}$, $\{4,6\}$ and $\{5,6\}$, which may be paired two ways. 
Considering the two elements of $\mathcal F$ on $\{1,2,3,4,5\}$, exactly one
must contain the cherry $\{1,2\}$. If the other cherry in this 5-taxon tree is
$\{3,5\}$, then this determines $T_1$ as $\{12|3456, 124|356, 1234|56\}$, and hence $T_2$ is determined as well.
Similarly, if the second cherry in the 5-taxon tree is $\{3,4\}$, then $T_1$ and $T_2$ are again uniquely determined.
If the second cherry is $\{4,5\}$, however, $T_1$ may be either $\{12|3456,123|456,1234|56\}$ or $\{12|3456,123|456,1235|46\}$.  
Considering the element of $\mathcal F$ on $\{1,2,3,4,5\}$ that contains cherry $\{1,3\}$, we likewise obtain two unique trees except
in the case where the second cherry is $\{4,5\}$, in which case $T_2$ could be either $\{13|2456, 123|456, 1234|56\}$ or
$\{13|2456,123|456,1235|46\}$. Finally, since only one of $T_1$ and $T_2$ can have cherry $\{5,6\}$, the only remaining
ambiguous case is that described in the statement of the Proposition.
\end{proof}

\begin{lemma}\label{lem:6inv}
  Consider the trees $T_1, T_2, T_3,$ and $T_4$ in
  equations \eqref{eq:2-tree-pairs} from Proposition
  \ref{prop:quartetmatched}.  Define the linear polynomial
$$f = q_{GGGGGG} + q_{GTTTTG} - q_{GTGGTG} - q_{GGTTGG}.$$
Then, for the JC and K2P models, $f$ satisfies
$$f \in \mathcal{I}(V_{T_1} \ast V_{T_2}) \setminus \mathcal{I}(V_{T_i}), \quad i \in \{3,4\}.$$
In particular, $V_{T_i} \not\subseteq V_{T_1} \ast V_{T_2}$, $i \in
\{3,4\}.$
\end{lemma}

\begin{proof}	
  By symmetry of the relationship between trees $T_3$ and $T_4$ to
  that of
  $T_1$ and $T_2$, it suffices to prove the statement in the case that
  $i = 3$.  First, we will show that $f \in \mathcal{I}(V_{T_1} \ast
  V_{T_2})$.  Denote by $a$'s and $b$'s the parameters of the trees
  $T_1$ and $T_2$, respectively.
One checks that
\begin{multline*}q_{GGGGGG} = \pi  a^1_{G} a^2_G a^3_G a^4_G a^5_G a^6_G a^{12|3456}_A a^{123|456}_G a^{1234|56}_A \\+ (1-\pi) b^1_{G} b^2_G b^3_G b^4_G b^5_G b^6_G b^{13|2456}_A b^{123|456}_G b^{1235|46}_A, \end{multline*}
\begin{multline*}q_{GTTTTG} = \pi  a^1_{G} a^2_T a^3_T a^4_T a^5_T a^6_G a^{12|3456}_C a^{123|456}_G a^{1234|56}_C \\+ (1-\pi) b^1_{G} b^2_T b^3_T b^4_T b^5_T b^6_G b^{13|2456}_C b^{123|456}_G b^{1235|46}_C, \end{multline*}
\begin{multline*}q_{GTGGTG} = \pi  a^1_{G} a^2_T a^3_G a^4_G a^5_T a^6_G a^{12|3456}_C a^{123|456}_T a^{1234|56}_C \\+ (1-\pi) b^1_{G} b^2_T b^3_G b^4_G b^5_T b^6_G b^{13|2456}_A b^{123|456}_T b^{1235|46}_A, \end{multline*}
\begin{multline*}q_{GGTTGG} = \pi  a^1_{G} a^2_G a^3_T a^4_T a^5_G a^6_G a^{12|3456}_A a^{123|456}_T a^{1234|56}_A \\+ (1-\pi) b^1_{G} b^2_G b^3_T b^4_T b^5_G b^6_G b^{13|2456}_C b^{123|456}_T b^{1235|46}_C. \end{multline*}
Recall that for the JC and K2P models, we have $a^e_G = a^e_T$ and $b^e_G = b^e_T$ for all $e$. Therefore,  $f \in \mathcal{I}(V_{T_1} \ast V_{T_2})$.

On the other hand, if we denote the parameters for the tree $T_3$ by $c$'s, we have that:
$$q_{GGGGGG} = c^1_{G} c^2_G c^3_G c^4_G c^5_G c^6_G c^{13|2456}_A c^{123|456}_G c^{1234|56}_A,$$ 
$$q_{GTTTTG} = c^1_{G} c^2_T c^3_T c^4_T c^5_T c^6_G c^{13|2456}_C c^{123|456}_G c^{1234|56}_C,$$
$$q_{GTGGTG} = c^1_{G} c^2_T c^3_G c^4_G c^5_T c^6_G c^{13|2456}_A c^{123|456}_T c^{1234|56}_C,$$
$$q_{GGTTGG} = c^1_{G} c^2_G c^3_T c^4_T c^5_G c^6_G c^{13|2456}_C c^{123|456}_T c^{1234|56}_A.$$
Although for the JC model $c^e_C=c^e_G = c^e_T$ and $c^e_A=1$ for all $e$,  the linear polynomial $f$ evaluated at the expressions above does not give the zero polynomial.  Therefore  $f \notin \mathcal{I}(V_{T_3})$ for the JC, and hence for the K2P, model. 
\end{proof}

\smallskip

Finally, we pull together all of the results in this section in the proof of the main Theorem on tree identifiability:

\begin{proof}[Proof of Theorem \ref{thm:maintwotree}]
If the trees relate only 4 taxa,
Corollary \ref{cor:quartets} provides the claim.  
By Theorem
\ref{thm:sixToInfty}, it is now enough to consider cases with $n=5,6$.

If $\{T_1,T_2\}$ and $\{T_3,T_4\}$ are not quartet-matched, then there
is a quartet $Q$ of taxa such that $\{T_1|_Q,T_2|_Q\}
\neq\{T_3|_Q,T_4|_Q\}$.  Thus by Corollary \ref{cor:quartets} and
Lemma \ref{lem:subtree}, the claim follows.

If $\{T_1,T_2\}$ and $\{T_3,T_4\}$ are quartet-matched, then by
Proposition \ref{prop:quartetmatched}, up to symmetry, we need only
consider the case described by equations \eqref{eq:2-tree-pairs}. But
then Lemma \ref{lem:6inv} implies
$$\mathcal{I}(V_{T_3} \ast V_{T_4}) \setminus \mathcal{I}(V_{T_1}\ast V_{T_2})$$
contains a linear invariant, so 
$$V_{T_1} \ast V_{T_2} \not\subseteq V_{T_3}\ast V_{T_4}.$$
\end{proof}


\section{Comparing $2$-tree mixtures with unmixed models}
\label{sec:mix-and-unmix}

In this section, we report on preliminary investigations on distinguishing unmixed models from 2-tree
mixtures. More precisely, we study the
following question: For which triples of trees $T_1, T_2, T_3$ is $V_{T_3}
\not\subseteq V_{T_1} \ast V_{T_2}$?  We have already used instances of this
in establishing generic identifiability of trees in the $2$-tree mixture model,  
but our earlier work does not yield a general answer.

That we can distinguish a single-class, unmixed model $V_{T_3}$ from a 2-tree mixture model
$V_{T_1} \ast V_{T_2}$, as long as $T_3$ is not too closely related to
$T_1$ and $T_2$ is easily shown, however.  Indeed, Lemma \ref{lem:linearinv} and a variant
of Lemma \ref{lem:subtree} imply:

\begin{prop}
  If there is a four-element set $Q \subseteq [n]$ such that $T_3|_Q
  \notin \{ T_1|_Q, T_2|_Q \}$ then $V_{T_3} \not\subseteq V_{T_1}
  \ast V_{T_2}$.
\end{prop}

The smallest instance of a tree $T_3$ all of whose quartet trees 
arise from $T_1$ and $T_2$ occurs with $n = 5$ leaves for the triple:
\begin{align} \label{eq:3-5}
T_1 &= \{12|345, 123|45\},\notag\\  
T_2 &= \{13|245,  134|25 \},\\ 
 T_3 & = \{123|45, 13|245 \}\notag.
\end{align} In fact, this example is unique up to the action of $\mathfrak S_5$ on leaf labels.

We performed a computation using the computer algebra program Singular
\cite{GPS09}, which rigorously verified the following:

\begin{thm}\label{thm:5leafbad}
For the three 5-taxon  trees $T_1, T_2, T_3$ in \eqref{eq:3-5}, under the JC model, 
$$V_{T_3} \subseteq V_{T_1} \ast V_{T_2}.$$ 
\end{thm}

\begin{proof}
We  explain the approach behind  our computation.  
  
All of the JC varieties $V_T$ for an $n$-leaf tree are invariant under
an action of the torus $(\mathbb{C}^*)^n$. This action arises from
rescaling the pendant edge parameters.  That is, if $q \in V_T$ and
$\lambda \in (\mathbb{C}^*)^n$, then for any subforest $F$ of $T$ the
Fourier coordinate $q_F$ is transformed as $(\lambda \cdot q)_F = q_F
\prod_{e \in L(F)} \lambda_e$, where $L(F)$ is the set of pendant
edges appearing in $F$.  Since $\lambda \cdot q \in V_T$, it suffices
to prove the claimed containment in the theorem in the case where all
pendant edge parameters on the tree $T_3$ are set to $1$.

Let $V$ and $W$ be two varieties.  Note that $V \subseteq W$, is
equivalent to $\mathcal{I}(W) \subseteq \mathcal{I}(V)$.  This
containment of ideals holds if, and only if, $\mathcal{I}(W) +
\mathcal{I}(V) = \mathcal{I}(V)$, which we use to speed up
computations.  Hence, it suffices to show that
\begin{equation} \label{eq:containsec}
\mathcal{I}(V_{T_1} \ast V_{T_2}) + \mathcal{I}(V) = \mathcal{I}(V)
\end{equation}
where $V$ is the subvariety of $V_{T_3}$ where all the pendant edge parameters are set to $1$.

Finally, though in principle it is possible to compute
$\mathcal{I}(V_{T_1} \ast V_{T_2})$ directly, it is beyond current
capabilities.  However, an alternative approach to join ideals uses
elimination: if $I, J \subseteq \mathbb{C}[q]$ are two ideals, their
join ideal is
$$I\ast J = (I(q') + J(q - q')) \cap \mathbb{C}[q] $$
where $I(q')$ is the ideal $I$ with variables $q_i'$ 
substituted for variables $q_i$, and $J(q- q')$ is ideal $J$ with $q_i
- q'_i$ substituted for $q_i$.  Hence, we can test
(\ref{eq:containsec}) by testing if
$$\mathcal{I}(V) = (\mathcal{I}(V_{T_1})(q') + \mathcal{I}(V_{T_2})(q - q') + \mathcal{I}(V)(q) ) \cap \mathbb{C}[q]. $$
This statement is verified by the code we provide in the supplementary materials \cite{TTwebsite}.
\end{proof}

\smallskip

Theorem \ref{thm:5leafbad} raises as many questions as it answers.
First, note that it is a statement about complex varieties, and leaves
open the possibility that ${\mathcal M}_{T_3} \not\subseteq {\mathcal
  M}_{T_1} \ast {\mathcal M}_{T_2}.$ We investigated this
computationally as follows, using code available in
\cite{TTwebsite}. Choosing random JC parameters on $T_3$, we
repeatedly produced a point in $V_{T_3}\subseteq V_{T_1}*V_{T_2}$,
thus obtaining a sample with high probability of exhibiting generic
behavior. For each such point, we then produced a system of algebraic
equations whose solutions would give mixture parameters on $T_1$ and
$T_2$ to produce this point.  The solution set forms an algebraic
variety, which in our trials was always of dimension 2. A primary
decomposition of the ideal showed there were three components of the
solution set, two of dimension 2 and one of dimension 1.

One of the 2-dimensional components was defined in part by setting one
internal edge length on $T_1$ to infinity and one internal edge length
on $T_2$ to 0. The mixing parameter, all split parameters in $T_2$,
and all but 4 split parameters in $T_1$ were uniquely determined. Two
quadratic relationships in 2 variables each held for the remaining
parameters.  The other 2-dimensional component is similar, with the
roles of $T_1$ and $T_2$ reversed. The 1-dimensional component
requires that an internal edge on each tree have length 0, but allows
the mixing parameter to vary along with two edges on each tree. (See
\cite{TTwebsite} for the precise results.)

It is worth highlighting that the only 2-tree mixtures matching the
1-tree distribution were of this extreme nature, with some internal
edges of length 0 or infinity. If one allows these values, then there
are instances of all mixture parameters being in a stochastically
meaningful range.  Of course formally establishing any conjectures
these calculations suggest would require a detailed semi-algebraic
analysis of these models.

A second question Theorem \ref{thm:5leafbad} might lead one to ask is
if $T_3$ is a tree all of whose quartets comes from either $T_1$ or
$T_2$, then is $V_{T_3} \subseteq V_{T_1} \ast V_{T_2}$? However, we
have already seen an instance where this failed in Lemma
\ref{lem:6inv}.  It would be interesting to characterize precisely
when these types of containments arise.

Finally, it is not at all clear if the containment in Theorem
\ref{thm:5leafbad} is a special phenomenon for the JC model, or if it
can occur more generally for other group-based or more general
phylogenetic models.  Answering such questions will require an
understanding of this phenomenon beyond the computational perspective.


\section{Identifiability of Continuous Parameters}
\label{sec:id-params}

Assuming tree topologies are already known, we next explore the
generic identifiability of the continuous parameters in group-based
mixture models. We use both rigorous arguments and computational
approaches to address this issue.
While standard laptop computers were sufficient for most of this work (see \cite{TTwebsite}),
the more intensive computations were performed on a more powerful machine
provided by Erich Kaltofen of NCSU.

\smallskip

Proving a model has identifiable continuous parameters requires
showing that the parameterization map is one-to-one. Without any
special assumptions on the map, it may be one-to-one on some region of
parameter space, but not on another.  Well-known to algebraic
geometers, however, is that parameterization maps defined by
polynomial formulas, such as those for the models we study, have the
nice feature that they exhibit a \emph{generic} behavior.
More
specifically, there is some $k\in \{1,2,3,\dots,\infty\}$ such that
for all parameter values except those in some exceptional set $E$, the
map will be $k$-to-one (\emph{cf.}, Prop.~7.16 of \cite{Harris}). Crucially, the set $E$ where the generic
behavior may fail is closed and of Lebesgue measure 0 within the full
parameter space (since it is a proper algebraic subvariety). 
In the case of complex univariate polynomials, this fact is more widely familiar:
given an $n$-th degree
polynomial  $p(z)$, for almost all $\alpha\in \mathbb C$ the equation $p(z)=\alpha$ has $n$ distinct roots. However, for a finite number of exceptional values of $\alpha$ there may be fewer distinct roots. Thus $p$ defines a generically $n$-to-one map from $\mathbb C$ to $\mathbb C$. 

One can computationally determine the generic behavior \emph{with high
  probability} as follows: For a specific choice of parameters,
calculate the cardinality of the set of all other choices of
parameters with the same image. If, for many such random choices, one
finds this fiber is of size $k$, there can be little doubt that the
map is generically $k$-to-1. These computations can be performed
exactly by computational algebra software such as Macaulay2 \cite{M2}
or Singular \cite{GPS09}, and carefully performed repeated trials can
give one high confidence. Of course, such an approach
does not rigorously establish results. However, the use of random data to reliably study behavior of specific polynomial equations
is not novel. For instance, Section 6 of \cite{HosKheSturm} gives a different application of the idea in phylogenetics.

This approach unfortunately does not give any quantifiable meaning to the term `high probability,' as we lack any explicit information on the set $E$ where non-generic behavior may arise. If a non-zero multivariate polynomial vanishing on $E$ were known, we would only need to compute that the map was $k$-to-one for a single point not satisfying that polynomial, and obtain a rigorous result. If we knew only the degree of such a polynomial, by the Schwartz-Zippel Theorem (\emph{cf.}, for instance, \cite{Rudich}), we could
produce points with arbitrarily small probability of lying in $E$, and use these to quantify our terminology. However, we have no such information, and thus our confidence in having determined the generic behavior is based partly on experience. In choosing points for calculations, a useful heuristic is to pick coordinates to be random rational numbers  (perhaps also requiring that they be expressible using disjoint sets of primes), in hopes that the unknown polynomial equations describing $E$ are less likely to
be satisfied. Indeed, if 25 points chosen in this way all produce the same value of $k$, while it is possible they all lie in $E$, the evidence is strong
that they do not.

We label statements with ``Theorem*'' or ``Proposition*'' if
we are only highly confident of them through such computation.
Unstarred statements are rigorously proved. Thus while we are careful
to distinguish between results with rigorous proof and those depending
on such calculations, we are highly confident of both.

\smallskip

One of the results we found computationally was a particularly
surprising non-identifiability result for continuous parameters of
4-taxon tree mixtures under the JC model. Nonetheless, passing to
5-taxon trees restores identifiability.

The first main result in this section is:

\begin{thmstar}\label{thm:paramident} 
  For the JC model, the continuous parameters in the $2$-tree
  mixture with parameterization $\psi_{T_1,T_2}$ are generically identifiable for binary trees with $n \geq
  4$ leaves, except in the case that $n = 4$ and $T_1 = T_2$.
\end{thmstar}

An issue that will arise in our proof of Theorem* \ref{thm:paramident} and related results, concerns the maps $\psi_T$ parametrizing $V_T$ for the JC, K2P, and K3P models. The proof in \cite{Chang} of the identifiability of numerical parameters for the general Markov model shows that identifiability of numerical parameters only holds up to permutation of states at internal nodes of the tree. Permuting the states at an internal node corresponds to permuting rows of transition matrices on edges leading out of the node, and columns of matrices on edges leading into the node. As any permutation of the rows or columns of a JC matrix that is also a JC matrix is identical to the original matrix, this implies the JC parameterization is generically one-to-one. For a generic K2P matrix, there are two orderings of the rows that have K2P form, and hence the parametrization map is generically $2^{n-2}$-to-one. For a generic K3P matrix, there are four orderings of the rows that have K3P form, and hence the parametrization map is generically $4^{n-2}$-to-one. To avoid complications in statements due to these understood failures of identifiability in its strictest sense, it is more convenient to focus on the $k$-to-oneness of the maps $\phi_{T_1,T_2}$, using
equation \eqref{eq:compparam} to relate results to $\psi_{T_1,T_2}$.

\smallskip

The first step toward Theorem* \ref{thm:paramident} is performing computations to establish the following.

\begin{propstar}\label{prop:4leafparamid1} 
Let $T_1\ne T_2$ be binary trees with four leaves.  Then for the JC model, the map 
$$\phi_{T_1, T_2} : V_{T_1} \times V_{T_2} \times \mathbb{P}^1 \dashrightarrow V_{T_1} \ast V_{T_2}$$
is generically one-to-one.
\end{propstar}

\begin{proof}[Calculation] 
  From randomly chosen rational parameters in the domain of
  $\psi_{T_1,T_2}$, we computed a point $p\in
  V_{T_1}*V_{T_2}$. Solving the system of polynomial equations
  $\psi_{T_1,T_2}(s_1,s_2,\pi)=p$ determines the (complex) preimage of
  $p$. This preimage can be calculated using Gr\"obner bases, and was
  found to consist of a single point for the many such random choices
  we made. We can be therefore be highly confident that
  $\psi_{T_1,T_2}$ is one-to-one, by the existence of a generic
  behaviour of any polynomial map. That $\phi_{T_1,T_2}$ is one-to-one
  then follows from the fact that $\psi_{T_1}$ and $\psi_{T_2}$ are
  generically one-to-one parameterizations of $V_{T_1}$ and $V_{T_2}$.

Code is provided in the supplementary materials \cite{TTwebsite}.
\end{proof}

\smallskip

Although we attempted to perform similar calculations to extend
Proposition* \ref{prop:4leafparamid1} to the K2P and K3P models,
these failed to terminate in 3 weeks time.

\smallskip

In the case of a mixture on two trees with the same topology, the
possibility of interchanging the mixture components shows the map
cannot be one-to-one. Generic identifiability thus corresponds to
generic two-to-oneness in this case. For this type of mixture, we are able to
perform computations for both the JC and Kimura models.

\begin{propstar}\label{prop:4leafparamid2}
Let $T$ be a binary tree with four leaves.  Then for the K2P and K3P models, the map
$$\phi_{T,T}: V_T \times V_T \times \mathbb{P}^1 \dashrightarrow V_{T} \ast V_T$$
is generically two-to-one.  For the JC model, the map $\phi_{T,T}$ is
generically twelve-to-one.  \end{propstar}

\begin{proof}[Calculation]
  The calculations which indicate this holds with high probability is
  similar to that for Proposition* \ref{prop:4leafparamid1}.

Code is provided in the supplementary materials \cite{TTwebsite}.
\end{proof}

\smallskip

Note that the twelve-to-oneness in the case of the JC model is not
merely a mathematical anomaly relevant to complex parameter choices
only.  This type of non-identifiability for secant parameters can and
does occur for stochastically meaningful parameters.

\begin{ex} Searches of parameter space give instances of $2$, $4$ or
  $8$ stochastic parameter choices producing the same image in the
  4-taxon 2-class JC mixture on a single tree topology. For instance,
  if $a^e=a^e_C=a^e_G=a^e_T$ denotes the JC parameters for one class
  on $T$, and $b^e=b^e_C=b^e_G=b^e_T$ those for the other class, and
  $\pi$ the proportion of the first class, then
\begin{gather*}\pi = 0.1,\\
 a^1 = 0.05,\ 
 a^2 = 0.10,\ 
 a^3 = 0.12,\ 
 a^4 = 0.04,\ 
 a^{12|34} = 0.01,\\
 b^1 = 0.04,\ 
 b^2 = 0.14,\ 
 b^3 = 0.10,\ 
 b^4 = 0.11,\ 
b^{12|34} = 0.46,
\end{gather*} 
and 7 other choices of parameters have the same image. Up to
interchanging classes, there are 4 essentially different choices.
Code verifying this example, and examples showing $2$ or $4$
biologically relevant preimages, are included in the supplementary
materials \cite{TTwebsite}.  We do not know if exactly $6$, $10$, or
$12$ biologically relevant preimages can occur.
\end{ex} 

We rigorously establish the following:

\begin{prop}\label{prop:fiveleafsametree}
  Let $T$ be a binary tree with five leaves.  Then for the JC, K2P,
  and K3P models the map
$$\phi_{T,T}: V_T \times V_T \times \mathbb{P}^1 \dashrightarrow V_{T} \ast V_T$$
is generically two-to-one.
\end{prop}

The proof of Proposition \ref{prop:fiveleafsametree} depends on a result of
J.~Kruskal concerning uniqueness of rank 1 tensor decompositions for
$3$-way arrays.  As this has been exploited elsewhere
\cite{Allman2009b,Allman2009a} to study identifiability of models, we
give only essentials here. If $M_1,M_2,M_3$ are three matrices with
$r$ rows, and $\boldsymbol \pi$ is an $r$-element vector, let $\mathbf
m^i_j$ denote row $i$ of matrix $M_j$. Let
$$[\boldsymbol \pi;M_1,M_2,M_3]=\sum_{i=1}^r \pi_i\mathbf  m^i_1\otimes \mathbf m^i_2 \otimes \mathbf m^i_3.$$
The form of Kruskal's theorem most useful for our purposes is the
following, from \cite{Allman2009b}.

\begin{thm}(Kruskal)\label{thm:kruskal} Let $\boldsymbol \pi$ be an
  $r$-element vector of non-zero numbers, and $M_1,M_2,M_3$ three
  matrices with $r$ rows, all of whose row sums are 1. Let $I_i$, the
  \emph{Kruskal rank} of $M_i$, be the largest integer such that every
  set of $I_i$ rows of $M_i$ is independent, and
  suppose $$I_1+I_2+I_3\ge 2r+2.$$ Then if $[\boldsymbol
  \pi;M_1,M_2,M_3]=[\boldsymbol \pi';M_1',M_2',M_3']$, there is a
  permutation $P$ such that
$$\boldsymbol \pi=P\boldsymbol \pi',\ \ M_1=PM_1',\ \ M_2=PM_2',\ \  M_3=PM_3'.$$
\end{thm}

\begin{proof}[Proof of Proposition \ref{prop:fiveleafsametree}]
  To fix notation, let $T$ have non-trivial splits $\{12|345,123|45\}$
  and let $\rho$ denote the internal node on the pendant edge leading to
  leaf 3. Denote by $\psi$ the natural parameterization of $V_T$,
  in terms of the entries of $4\times 4$ Markov matrices.  As there is no
  advantage to working in Fourier coordinates here, we use standard
  ones for $V_T$ and $V_T*V_T$.

  With $id$ the identity map on $\mathbb P^1$, it is enough to show
  the parameterization map $\phi_{T,T}\circ (\psi\times\psi\times id)$
  of $V_T*V_T$ is generically two-to-one.
  
  The map $\phi_{T,T}\circ (\psi\times\psi \times id)$ can be made
  explicit as follows: Root the tree at $\rho$, and assign stochastic
  matrices to the edges of the tree giving conditional probabilities
  of state changes along those edges.  If $\pi$ is the mixing
  parameter, and $\mathbf u=(1/4,1/4,1/4,1/4)$, then an $8$-element
  vector $$\boldsymbol \pi=(\pi\mathbf u,(1-\pi)\mathbf u),$$ gives
  the state distribution at the root. On the pendant edge leading to
  leaf $i$, a $8\times 4$ matrix $M_i$ composed of two stacked Markov
  matrices $M_i^{(1)}$ and $M_i^{(2)}$ of the appropriate form gives
  conditional probabilities, while on the internal edges there are
  $8\times 8$ block-diagonal matrices $M_{12|345}$ and $M_{123|45}$
  with two $4\times 4$ blocks $M_{12|345}^{(1)}$, $M_{12|345}^{(2)}$
  and $M_{123|45}^{(1)}$, $M_{123|45}^{(2)}$ of an appropriate
  form. Thus the superscript $(1)$ or $(2)$ refers to the class in the
  mixture.

  Now the $8\times 16$ matrix
  $M_{12}=M_{12|345}(M_1\otimes^{row}M_2)$, where $\otimes^{row}$
  denotes tensor products of corresponding rows, gives probabilities
  of observing pairs of states at leaves 1 and 2 conditioned on the
  state at $\rho$. A similar $8\times 16$ matrix product $M_{34}$ gives
  probabilities of observing pairs of states at leaves 4 and 5
  conditioned on the state at $\rho$.

  For any choice of $\boldsymbol \pi$ and the $M_e^{(1)},M_e^{(2)}$,
  the image $X$ under $\phi_{T,T}\circ (\psi\times\psi\times id)$ has
  the same entries as the 3-way array $[\boldsymbol
  \pi;M_{12},M_3,M_{45}]$. But for generic choices of parameters, one
  can check that $M_{12}$ and $M_{45}$ have Kruskal rank 8, and $M_3$
  has Kruskal rank $\ge2$.  Indeed, one need only check that this
  holds for a single choice of the parameters, since then the
  condition, which is defined by polynomial inequalities, can fail
  only on a proper subvariety of the parameter space. For instance,
  choosing $M_1^{(1)}=M_2^{(1)}$ to be a JC matrix with off-diagonal
  entry $0.1$, $M_1^{(2)}=M_2^{(2)}=I_4$, and $M_{12|345}=I_8$, a
  calculation shows $M_{12}$ has rank 8, and hence Kruskal rank 8.

  Applying Theorem \ref{thm:kruskal}, we get that $\boldsymbol \pi$,
  $M_{12}$, $M_3$, and $M_{45}$ are all uniquely determined up to
  simultaneous permutation of rows.  
  
  However, because of the special
  form of the Markov matrices for the models we consider,
  for generic
  JC parameters there are exactly two orderings to the rows of $M_3$ so
  that it is two stacked blocks of the correct form, and these differ
  by simply interchanging the blocks. Thus we may recover $M_3^{(1)}$
  and $M_3^{(2)}$, up to order. Fixing the ordering of the rows of
  $M_3$ so $M_3^{(1)}$ is on top fixes an ordering of the rows of
  $\boldsymbol \pi$, $M_{12|345}$, and $M_{123|45}$ as well. Letting a
  superscript of 1 denote the top 4 rows, and a superscript of 2 the
  bottom 4 rows of an 8-row matrix, the mixture distribution can be
  written as $$\pi\left [\mathbf
    u;M_{12}^{(1)},M_3^{(1)},M_{45}^{(1)}\right]+(1-\pi)\left[\mathbf
    u;M_{12}^{(2)},M_3^{(2)},M_{45}^{(2)}\right].$$ But this weighted
  sum is simply the weighted sum of the two points in the image of
  $\psi$ corresponding to the two classes. As $\psi$ is known to be
  generically one-to-one, the parameterization of $V_T*V_T$ is
  two-to-one in the JC case.

As discussed following the statement of Theorem* \ref{thm:paramident}, for the K2P and K3P models there are additional orderings of the rows of $M_3$ so that is is two stacked blocks of the correct form. Regardless of which ordering we choose, however, by arguing as in the preceding paragraph we are led to the same two points
  in the image of $\psi_T$. Thus for these models also we see the parameterization of $V_T*V_T$ by $\phi_{T,T}$ is
  two-to-one.
\end{proof}

\smallskip

Note that the use of Kruskal's theorem in this proof extends to a
2-class CFN mixture model on a 5-taxon tree, as then the Kruskal ranks
of the matrices $M_{12}$ and $M_{45}$ are generically 4, while $M_2$
has Kruskal rank $\ge 2$ .  Although we do not focus on that model
here, we record the result, as it is helps place the examples of
\cite{Matsen2007} for 4-taxon trees into context.

\begin{prop} Let $T$ be a binary tree with five leaves.  Then for the
  CFN models the map
$$\phi_{T,T}: V_T \times V_T \times \mathbb{P}^1 \dashrightarrow V_{T} \ast V_T$$
is generically two-to-one.
\end{prop}

\begin{propstar}\label{lem:5taxa}
  For the JC $2$-tree mixture model on 5-taxon binary trees,
  the continuous parameters are generically identifiable.
\end{propstar}

\begin{proof}
  If $T_1$ and $T_2$ have no cherries in common, then all their
  induced quartet trees disagree.  Thus applying Proposition*
  \ref{prop:4leafparamid1} to all 4-taxon marginalizations shows all
  parameters are generically identifiable.

  If $T_1$ and $T_2$ have 2 cherries in common, they are identical,
  and Proposition \ref{prop:fiveleafsametree} gives the claim.

  If $T_1$ and $T_2$ have a single cherry in common, we may assume
  they are $T_1=\{12|345,123|45\}$ and $T_2=\{12|345,124|35\}$. Also,
  since the parameters are generic, we may assume the mixing parameter
  giving the class size for the $T_1 $ component is $\pi\ne 1/2$.
  Then marginalizing to the taxa $\{1,3,4,5\}$ and applying
  Proposition* \ref{prop:4leafparamid1} identifies the parameters on 4
  edges of each of the trees, as well as the class size $\pi$ for
  $T_1$.  

 Marginalizing to quartets involving taxa 1 and 2, and applying Proposition*
    \ref{prop:4leafparamid2} to them, there are 12 points in a generic
    fiber.  However, such a generic fiber will have 6 distinct pairs
    of values $\{\pi, 1-\pi\}$, and we use the value of the mixing
    parameter $\pi$ determined above to match parameters with $T_1$
    and $T_2$.
\end{proof}

\smallskip

The results above allow us to argue for the generic identifiability of
parameters claimed in Theorem* \ref{thm:paramident}.

\smallskip

\begin{proof}[Proof of Theorem* \ref{thm:paramident}] 
The $n = 4$ case is Proposition*s \ref{prop:4leafparamid1},
and the $n=5$ case is Proposition* \ref{lem:5taxa}.

  For $n> 5$ leaves, by assuming that the parameters are generic we
  may also suppose the mixing parameter $\pi\ne 1/2$.
  
  By marginalizing to 5-taxon subsets, and applying Proposition*
  \ref{lem:5taxa}, we may identify parameters on each pair of induced
  5-taxon trees, but we must determine which come from which tree. If
  there is at least one 5-taxon subset for which $T_1$ and $T_2$
  induce different subtrees, then we know the class size parameter
  $\pi$ for $T_1$. Using this known value, we can determine which
  induced 5-taxon parameters arise from $T_1$ and which arise from
  $T_2$, even when the 5-taxon subtrees are topologically the same. If
  all 5-taxon subtrees of $T_1$ and $T_2$ agree, so $T_1=T_2$, then we
  instead use the value of $\pi_1$ to collect 5-taxon subtree
  parameters from each copy of the tree. As the parameters for $T_1$
  and $T_2$ are elements of the collection of induced parameters, we
  thus identify all parameters on the full trees.
\end{proof}

\smallskip

  In closing, note that the arguments in the proof of Theorem*
  \ref{thm:paramident} in combination with the results of Proposition
  \ref{prop:fiveleafsametree},  rigorously prove the following result, in the case
  of identical tree topologies.
  
\begin{thm}
  For the JC, K2P, and K3P models, the continuous parameters in the
  $2$-tree mixture on the same tree topology are generically identifiable for
  binary trees with $n \geq 5$ leaves.
\end{thm}


\bibliographystyle{plain}

\bibliography{2tree}

\vfill

\end{document}